\documentclass[a4paper,12pt]{article}

\usepackage{amsmath,amssymb}
\usepackage{amscd}
\usepackage{mathrsfs}
\usepackage{amsthm}
\usepackage{color}

\usepackage{comment}

\theoremstyle{plain}

\newtheorem{theorem}{\bf Theorem}[section]

\newtheorem{corollary}[theorem]{\bf Corollary}

\theoremstyle{definition}

  \makeatletter
  \newcommand{\subsubsubsection}{\@startsection{paragraph}{4}{\z@}%
    {1.0\Cvs \@plus.5\Cdp \@minus.2\Cdp}%
    {.1\Cvs \@plus.3\Cdp}%
    {\reset@font\sffamily\normalsize}
  }
  \makeatother
\title{An index theorem for split-step quantum walks}

\author{Yasumichi Matsuzawa\thanks{Department of Mathematics, Faculty of Education, Shinshu University,
        6-Ro, Nishi-nagano, Nagano 380-8544, Japan,
        e-mail: myasu@shinshu-u.ac.jp}}
\date{}

\begin{document}
\maketitle

\begin{abstract}
Split-step quantum walks are models of supersymmetric quantum walk, and thus their Witten indices can be defined.
We prove that the Witten index of a split-step quantum walk coincides with the difference between the winding numbers 
of functions corresponding to the right-limit of coins and the left-limit of coins.
As a corollary, we give an alternative derivation of the index formula for split-step quantum walks, which is recently obtained by Suzuki and Tanaka.
\end{abstract}

\section{Introduction and main result}
Suzuki \cite{Su} introduced the notion of supersymmetric quantum walk in an abstract way.
Such a quantum walk is defined by a pair of two unitary self-adjoint operators  $\Gamma$ and $C$ on a Hilbert space $\mathscr{H}$.
The time evolution of the system is described by the unitary operator $U:=\Gamma C$.
Then $U$ has a chiral symmetry $\Gamma$, that is $\Gamma U\Gamma=U^*$.
Thus the self-adjoint operator $Q:={\rm Im}(U)=(U-U^*)/2i$, which is called a supercharge,
satisfies $\Gamma Q+Q\Gamma=0$.
Since $\Gamma$ is unitary self-adjoint, 
we can decompose $\mathscr{H}$ and $Q$ as
\[
\mathscr{H}=\ker{(\Gamma-1)}\oplus\ker{(\Gamma+1)},\ \ \ \ \ 
Q=\begin{pmatrix}
0&Q_+^*\\
Q_+&0
\end{pmatrix},
\]
where $Q_+$ is a bounded operator from $\ker{(\Gamma-1)}$ to $\ker{(\Gamma+1)}$.
Suzuki defined the Witten index of the pair $(\Gamma,C)$ by
\[
{\rm ind}(\Gamma, C) := \dim{\ker{Q_+}}-\dim{\ker{Q_+^*}}.
\]
If $Q_+$ is a Fredholm operator, then the pair $(\Gamma, C)$ is called a Fredholm pair.
In this case, the Witten index is nothing but the Fredholm index of $Q_+$.
We denote by ${\rm index}\,A$ the Fredholm index of a Fredholm operator $A$.
Then we have
\[
{\rm ind}(\Gamma, C) = {\rm index}\,Q_+,
\]
whenever $(\Gamma, C)$ is a Fredholm pair.

Very recently, Suzuki and Tanaka \cite{ST} have computed the Witten index of a split-step quantum walk \cite{KRBD}.
A split-step quantum walk is a pair $(\Gamma, C)$ of two unitary self-adjoint operators  
$\Gamma$ and $C$ on the Hilbert space $\ell^2(\mathbb{Z})\oplus\ell^2(\mathbb{Z})$ defined as follows.
Let $L$ be the left-shift operator on $\ell^2(\mathbb{Z})$, that is, 
\[
(L\psi)(x):=\psi(x+1),\ \ \ \ \ \psi\in\ell^2(\mathbb{Z}),\ \ x\in\mathbb{Z}.
\]
Let $a:\mathbb{Z}\to\mathbb{R}$, $b:\mathbb{Z}\to\mathbb{C}$ be functions such that $a(x)^2+|b(x)|^2=1$ for each $x\in\mathbb{Z}$.
We identify these functions with the corresponding multiplication operators on $\ell^2(\mathbb{\mathbb{Z}})$.
Then the shift operator $\Gamma$ and the coin operator $C$ of the system are defined by
\[
\Gamma:=
\begin{pmatrix}
p&qL\\
\bar{q}L^*&-p
\end{pmatrix},\ \ \ \ \ 
C:=\begin{pmatrix}
a&b^*\\
b&-a
\end{pmatrix},
\]
where $(p,q)\in\mathbb{R}\times\mathbb{C}$ are scalars satisfying $p^2+|q|^2=1$.
This definition is a generalization of Kitagawa's split-step quantum walk \cite{KRBD}.
For details, see \cite[Example 2.2]{FFS1} or \cite[Section 2]{FFS2}.
Since $\Gamma$ and $C$ are unitary self-adjoint, the pair $(\Gamma, C)$ defines a supersymmetric quantum walk.
In this paper, we suppose that the limits
\[
a(\pm\infty):=\lim_{x\to\pm\infty}a(x),\ \ \ \ \ b(\pm\infty):=\lim_{x\to\pm\infty}b(x)
\]
exist.
Then Suzuki and Tanaka have shown the following index formula \cite[Theorem A]{ST}.
\begin{theorem}[Suzuki-Tanaka index formula]
The pair $(\Gamma, C)$ is Fredholm if and only if $|p|\not=|a(+\infty)|$ and $|p|\not=|a(-\infty)|$.
In this case, we have
\[
{\rm ind}(\Gamma, C) = \begin{cases}
{\rm sgn}(p),\ \ \ \ \ &|a(+\infty)|<|p|<|a(-\infty)|,\\
-{\rm sgn}(p), &|a(-\infty)|<|p|<|a(+\infty)|,\\
0, &{\rm otherwise}.
\end{cases}
\]
\end{theorem}

Their key idea to compute the Witten index is the following decomposition of the supercharge $Q$ \cite[Theorem 6]{ST}.

\begin{theorem}[Suzuki-Tanaka decomposition]
There exists a unitary operator $\epsilon$ on $\ell^2(\mathbb{Z})\oplus\ell^2(\mathbb{Z})$ such that
\[
\epsilon^*\Gamma\epsilon=\begin{pmatrix}
1&0\\
0&-1
\end{pmatrix},\ \ \ \ \ 
\epsilon^* Q\epsilon= \begin{pmatrix}
0&Q_{\epsilon,+}^*\\
Q_{\epsilon,+}&0
\end{pmatrix}, 
\]
where the bounded operator $Q_{\epsilon,+}:\ell^2(\mathbb{Z})\to\ell^2(\mathbb{Z})$ is defined by
\[
Q_{\epsilon,+}: = \frac{i}{2}\left[(1+p)e^{i\theta}Lb-(1-p)e^{-i\theta}b^*L^*+|q|\{a(\cdot+1)+a\}\right]
\]
with
\[
\theta:=\begin{cases}\arg{q},\ \ \ \ \ &q\not=0,\\
0,&q=0.
\end{cases}
\]
Moreover, the pair $(\Gamma, C)$ is Fredholm if and only if the operator $Q_{\epsilon,+}$ is Fredholm.
In this case, we have
\[
{\rm ind}(\Gamma, C) = {\rm index}\,Q_{\epsilon,+}.
\]
\end{theorem}

Motivated by the above decomposition, we establish a connection between the Witten index and a topological index.
Before going to state the main result of the present paper, let us observe the asymptotic behavior of $Q_{\epsilon,+}$ as $x\to\pm\infty$.
Let $\mathbb{T}:=\{z\in\mathbb{C} \mid |z|=1\}$ be the one-dimensional torus.
Define functions $F_{\pm}:\mathbb{T}\to\mathbb{C}$ by
\[
F_{\pm}(z):=\frac{i}{2}\left\{(1+p)e^{i\theta}b(\pm\infty)z-(1-p)e^{-i\theta}\overline{b(\pm\infty)}\bar{z}+2|q|a(\pm\infty)\right\}.
\]
Then formally, we have
\[
Q_{\epsilon,+} \sim F_{\pm}(L),\ \ \ \ \ x\to\pm\infty.
\]
We are now ready to state our main result of this paper.
\begin{theorem}\label{main}
The pair $(\Gamma, C)$ is a Fredholm pair if and only if $F_+$ and $F_-$ vanish nowhere.
In this case, we have
\[
{\rm ind}(\Gamma, C) = {\rm wn}(F_+)-{\rm wn}(F_-),
\]
where ${\rm wn}(f)$ denotes the winding number of a continuous function $f:\mathbb{T}\to\mathbb{C}$ vanishing nowhere.
\end{theorem}

Our main result asserts that 
the analytic index (the Fredholm index) coincides with the topological index (the difference between the winding numbers).
On the other hand, Suzuki proved that the absolute value of the Witten index gives 
a lower bound on the sum of  the dimensions of the eigenspaces 
of $U=\Gamma C$ corresponding to the eigenvalues $\pm 1$ \cite[Theorem 3.4 (ii)]{Su}.
Hence, we arrive at the following corollary.

\begin{corollary}
Suppose that $F_+$ and $F_-$ vanish nowhere.
If ${\rm wn}(F_+)\not={\rm wn}(F_-)$, then $U$ has at least one eigenstate corresponding to the eigenvalue $1$ or $-1$.
\end{corollary}

The idea behind the proof of Theorem \ref{main} is to cut the space $\mathbb{Z}$ at $x=0$ into two half-spaces.
The cut can be realized as a compact perturbation of $Q_{\epsilon,+}$.
Note that Fredholm properties are invariant under compact perturbations. 
Moreover, after the cut, $Q_{\epsilon,+}$ becomes the direct sum of two Toeplitz operators with symbols $F_+(\bar{\cdot})$ and $F_-$,
thus we can apply the index theorem for Toeplitz operators to get the desired result.
More details will be discussed in Section 2 and Section 3.

We compare our results with a related work \cite{CGGSVWW}. 
Cedzich et al. defined an index for one-dimensional quantum walks with symmetries, including split-step quantum walks, from a representation theoretic point of view.
They showed that their index coincides with the difference between the winding numbers of functions corresponding to the right-limit and the left-limit of the time evolution operator.
Their index, however, is not a Fredholm index, and hence is different from a Witten index in the sense of Suzuki.
Moreover, their chiral symmetry differs from one used in the present paper.

This paper is organized as follows.
In section 2, we show a criterion corresponding to the above argument in a more general setting.
In section 3, we prove the main result by applying the criterion.
In section 4, we give an alternative derivation of the Suzuki-Tanaka index formula by using the main result.

\section{A general criterion}
Let $\{|x\rangle\mid x\in\mathbb{Z}\}$ be the standard orthogonal basis of $\ell^2(\mathbb{Z})$, 
and let $C(\mathbb{T})$ be the Banach space of continuous functions from $\mathbb{T}$ to $\mathbb{C}$. 
For any $f,g\in C(\mathbb{T})$, we define an operator $A(f,g)$ on $\ell^2(\mathbb{Z})$ by
\[
A(f,g)|x\rangle :=\begin{cases}g(L)|x\rangle, \ \ \ \ \ &x\geq 0,\\
f(L)|x\rangle, &x\leq -1,
\end{cases}
\]
where $L$ is the left-shift operator.
In this section, we study the Fredholmness of $A(f,g)$ and its Fredholm index.
Let $P_{\geq0}$ be the orthogonal projection onto the closed subspace spanned by $\{|x\rangle\mid x\geq0\}$,
and let $P_{\leq-1}$ be the orthogonal projection onto the closed subspace spanned by $\{|x\rangle\mid x\leq-1\}$.
We endow $C(\mathbb{T})\oplus C(\mathbb{T})$ with the norm
\[
\|(f,g)\|_{\infty} :=\max{(\|f\|_{\infty},\|g\|_{\infty})}.
\]
Then $C(\mathbb{T})\oplus C(\mathbb{T})$ is a Banach space.
For any $\psi\in\ell^2(\mathbb{Z})$ with finite support, we have
\begin{align*}
\|A(f,g)\psi\|^2 &\leq 2\|A(f,g)P_{\geq0}\psi\|^2+2\|A(f,g)P_{\leq-1}\psi\|^2\\
&= 2\|g(L)P_{\geq0}\psi\|^2+2\|f(L)P_{\leq-1}\psi\|^2\\
&\leq 2\|(f,g)\|_{\infty}^2\|\psi\|^2,
\end{align*}
thus $A(f,g)$ is bounded and $\|A(f,g)\|\leq\sqrt{2}\|(f,g)\|_{\infty}$ holds.

\begin{theorem}\label{criterion}
Let $f,g\in C(\mathbb{T})$ be arbitrary.
Then $A(f,g)$ is a Fredholm operator if and only if $f$ and $g$ vanish nowhere.
In this case, we have
\[
{\rm index}\,A(f,g) = {\rm wn}(g)-{\rm wn}(f).
\]
\end{theorem}

\begin{proof}
Define an operator $B(f,g)$ on $\ell^2(\mathbb{Z})$ by
\[
B(f,g)|x\rangle :=\begin{cases}P_{\geq0}g(L)|x\rangle, \ \ \ \ \ &x\geq 0,\\
P_{\leq-1}f(L)|x\rangle, &x\leq -1.
\end{cases}
\]
By the same argument as $A(f,g)$, we can show that $B(f,g)$ is bounded and $\|B(f,g)\|\leq\sqrt{2}\|(f,g)\|_{\infty}$ holds.

{\bf Step 1.} $A(f,g)-B(f,g)$ is a compact operator.\\
To see this, we note that the maps
\[
C(\mathbb{T})\oplus C(\mathbb{T})\ni (f,g)\to A(f,g),\ \ \ \ \ C(\mathbb{T})\oplus C(\mathbb{T})\ni (f,g)\to B(f,g)
\] 
are linear, and are continuous in the operator norm topology.
Since the set of compact operators is norm-closed, and since the set of trigonometric polynomials is dense in $C(\mathbb{T})$,
we may assume that $f$ and $g$ are trigonometric polynomials.
Write them as
\[
f(z)=\sum_{|m|\leq M}\alpha_mz^m,\ \ \ \ \ g(z)=\sum_{|m|\leq M}\beta_mz^m,\ \ \ \ \ z\in\mathbb{T},
\]
where $M\in\mathbb{N}$ and $\alpha_m,\beta_m\in\mathbb{C}$ for every $|m|\leq M$.
Then we have
\[
A(f,g)-B(f,g) = \sum_{|m|\leq M}D_m,\ \ \ \ \ D_m:=A(\alpha_mz^m,\beta_mz^m)-B(\alpha_mz^m,\beta_mz^m).
\]
Since $D_m|x\rangle=0$ for each $|x|\geq M+1$, $D_m$ is a finite rank operator, thus so is $A(f,g)-B(f,g)$.
This finishes the proof of Step 1.

Recall that the Hardy space $H^2$ is the closed subspace of $L^2(\mathbb{T},d\mu)$ spanned by $\{z^n \mid n\in\mathbb{Z},\ n\geq0\}$,
where $\mu$ denotes the normalized arc length measure on $\mathbb{T}$.
Let $P$ be the orthogonal projection of $L^2(\mathbb{T},d\mu)$ onto $H^2$.
For $h\in C(\mathbb{T})$, we identify $h$ with the corresponding multiplication operator on $L^2(\mathbb{T},d\mu)$.
The Toeplitz operator $T_h$ with symbol $h$ is a bounded operator on $H^2$ defined by $Ph$.
The index theorem for Toeplitz operators asserts that $T_h$ is Fredholm if and only if $h$ vanishes nowhere \cite[Corollary 3.5.12]{Mu}.
In this case, we have ${\rm index}\,T_h=-{\rm wn}(h)$ \cite[Theorem 3.5.15]{Mu}.

{\bf Step 2.} $B(f,g)$ is unitarily equivalent to $T_{g(\bar{\cdot})}\oplus T_{f}$.\\
To see this, we define a unitary operator $U:\ell^2(\mathbb{Z})\to H^2\oplus H^2$ by
\[
U|x\rangle := \begin{cases}\left(z^x,0\right),\ \ \ \ \ &x\geq0,\\
\left(0,z^{-x-1}\right),&x\leq-1.
\end{cases}
\]
By a limiting argument similar to Step 1,
we obtain $UB(f,g)U^*=T_{g(\bar{\cdot})}\oplus T_{f}$.
This finishes the proof of Step 2.

Since the Fredholmness and the Fredholm index are invariant under compact perturbations and unitary equivalence,
$A(f,g)$ is Fredholm if and only if $T_{g(\bar{\cdot})}\oplus T_{f}$ is Fredholm,
which is equivalent to that $f$ and $g$ vanish nowhere.
In this case, we have
\[
{\rm index}\,A(f,g) = {\rm index}\,T_{g(\bar{\cdot})}\oplus T_{f} = {\rm wn}(g)-{\rm wn}(f),
\]
where the last equality follows from the index theorem for Toeplitz operators.
This completes the proof of Theorem \ref{criterion}.
\end{proof}

\section{Proof of Theorem \ref{main}}
Let
\[
C(\pm\infty):=\begin{pmatrix}
a(\pm\infty)&\overline{b(\pm\infty)}\\
b(\pm\infty)&-a(\pm\infty)
\end{pmatrix},\ \ \ \ \ \tilde{C}(x):=\begin{cases}C(+\infty),\ \ \ \ \ &x\geq0,\\
C(-\infty),&x\leq-1.
\end{cases}
\]
We define a unitary self-adjoint operator $\tilde{C}$ on $\ell^2(\mathbb{Z})\oplus\ell^2(\mathbb{Z})=\ell^2(\mathbb{Z};\mathbb{C}^2)$ by
\[
(\tilde{C}\psi)(x):=\tilde{C}(x)\psi(x),\ \ \ \ \ \psi\in\ell^2(\mathbb{Z};\mathbb{C}^2),\ \ x\in\mathbb{Z}.
\]
Since $C-\tilde{C}$ is a compact operator, the pair $(\Gamma, C)$ is Fredholm if and only if $(\Gamma, \tilde{C})$ is Fredholm.
In this case, ${\rm ind}(\Gamma, C)={\rm ind}(\Gamma, \tilde{C})$ holds \cite[Remark 4.1]{Su}.
Hence, without loss of generality, we may assume that $C=\tilde{C}$.
Then, since $\{Q_{\epsilon,+}-A(F_-,F_+)\}|x\rangle=0$ for any $x\not=-1$, 
the operator $Q_{\epsilon,+}-A(F_-,F_+)$ is of finite rank.
In particular, it is compact.
Theorem \ref{main} now follows from Theorem \ref{criterion}.

\section{An alternative derivation of the Suzuki-Tanaka index formula}
We first remark that $F_{\pm}$ vanish nowhere if and only if $zF_{\pm}$ vanish nowhere.
In this case,
\[
{\rm wn}(zF_+)-{\rm wn}(zF_-) = \{{\rm wn}(F_+)+1\}-\{{\rm wn}(F_-)+1\} ={\rm wn}(F_+)-{\rm wn}(F_-)
\]
holds.
We regard $zF_{\pm}$ as polynomial functions defined on the whole $\mathbb{C}$.
If $|p|=1$, then the formula immediately follows. 
In what follows, we assume that $|p|\not=1$.

If $|a(+\infty)|=1$, then $zF_+$ vanishes nowhere on $\mathbb{T}$ and ${\rm wn}(zF_+)=1$ holds.
If $|a(+\infty)|\not=1$, then the zeros of the polynomial $zF_+$ are given by
\[
z_1^{(+)} := \frac{|q|\{1-a(+\infty)\}}{(1+p)e^{i\theta}b(+\infty)},\ \ \ \ \ 
z_2^{(+)} := -\frac{|q|\{1+a(+\infty)\}}{(1+p)e^{i\theta}b(+\infty)}.
\]
Thus $zF_+$ vanishes nowhere on $\mathbb{T}$ if and only if $|z_1^{(+)}|\not=1$ and $|z_2^{(+)}|\not=1$.
A direct computation shows that the last condition is equivalent to that $|p|\not=|a(+\infty)|$.
In this case, to determine the winding numbers of $zF_{\pm}$,
let us regard $\mathbb{T}$ as a circle oriented counterclockwise.
Then
\begin{align*}
{\rm wn}(zF_+) &= \frac{1}{2\pi i}\int_{\mathbb{T}}\frac{(zF_+)'(z)}{zF_+(z)}dz\\
&= \textrm{the number of zeros of $zF_+$ in the open unit disc},
\end{align*}
where the second equality follows from the fact that $zF_+$ are entire functions.
Since $|z_j^{(+)}|<1$ if and only if $(-1)^ja<p$ for each $j=1,2$, 
by considering the three cases $p>0$, $p=0$ and $p<0$ separately, we obtain
\[
{\rm wn}(zF_+)=\begin{cases}2,\ \ \ \ \ &|a(+\infty)|<p,\\
1,&|p|<|a(+\infty)|,\\
0,&|a(+\infty)|<-p.
\end{cases}
\]
Note that the above equality is true even for the case $|a(+\infty)|=1$.

Similarly, we can show that $zF_-$ vanishes nowhere on $\mathbb{T}$ if and only if $|p|\not=|a(-\infty)|$.
In this case, we have
\[
{\rm wn}(zF_+)=\begin{cases}2,\ \ \ \ \ &|a(-\infty)|<p,\\
1,&|p|<|a(-\infty)|,\\
0,&|a(-\infty)|<-p.
\end{cases}
\]
Therefore, by Theorem \ref{main}, we get the Suzuki-Tanaka index formula.

\section*{Acknowledgement}
I would like to thank Akito Suzuki and Yohei Tanaka for helpful discussions on split-step quantum walks and their careful reading of the first draft of this paper.
I also would like to thank Itaru Sasaki for informing me of the relation between the winding number of a holomorphic function and the number of zeros of it in the open unit disc, 
which was used in Section 4.

\end{document}